\documentclass[article,lineno]{biometrika}

\usepackage{amsmath,multibib}
\usepackage{enumitem}
\usepackage{times}
\usepackage{bm}
\usepackage{natbib}
\usepackage{amsfonts}
\usepackage[plain,noend]{algorithm2e}
\usepackage{mathrsfs}
\usepackage{graphicx}
\usepackage{url}
\usepackage{verbatim}
\usepackage{booktabs,tabularx}
\usepackage{multirow}
\usepackage{xcolor}
\usepackage{caption}
\usepackage{setspace}
\usepackage{tcolorbox}

\usepackage{deep-macro}

\newcolumntype{Y}{>{\centering\arraybackslash}X}
\newcolumntype{f}{>{\centering\arraybackslash}X}
\newcolumntype{h}{>{\hsize=.5\hsize\centering\arraybackslash\extracolsep{.1em}}X}

\newcolumntype{C}[1]{>{\hsize=#1\hsize\centering\arraybackslash}X}

\graphicspath{ {LPKsample/},}
\makeatletter
\renewcommand{\algocf@captiontext}[2]{#1\algocf@typo. \AlCapFnt{}#2} % text of caption
\def\@algocf@capt@plain{top}
\renewcommand{\algocf@makecaption}[2]{%
  \addtolength{\hsize}{\algomargin}%
  \sbox\@tempboxa{\algocf@captiontext{#1}{#2}}%
  \ifdim\wd\@tempboxa >\hsize%     % if caption is longer than a line
    \hskip .5\algomargin%
    \parbox[t]{\hsize}{\algocf@captiontext{#1}{#2}}% then caption is not centered
  \else%
    \global\@minipagefalse%
    \hbox to\hsize{\box\@tempboxa}% else caption is centered
  \fi%
  \addtolength{\hsize}{-\algomargin}%
}
\makeatother

\def\Bka{{\it Biometrika}}

\newcommand{\bthm}{\begin{theorem}}
\newcommand{\ethm}{\end{theorem}}
\newcommand{\bpf}{\begin{proof}}
\newcommand{\epf}{\end{proof}}
\newcommand{\tdeg}{\texttt{deg}}

\begin{document}

\jname{Biometrika}
\endpage{17}

\markboth{S. MUKHOPADHYAY \and K. WANG}{Nonparametric  High-dimensional K-sample Comparison}

\title{A Nonparametric Approach to High-dimensional\\k-sample Comparison Problems}

\author{SUBHADEEP MUKHOPADHYAY}
\affil{Stanford University, Department of Statistics, Stanford, CA 94305, USA \email{deep@unitedstatalgo.com}}

\author{and KAIJUN WANG}
\affil{Fred Hutchinson Cancer Research Center, Seattle, WA 98109, USA \email{kaijunwang.19@gmail.com}}

\maketitle

\begin{abstract}
High-dimensional k-sample comparison is a common applied problem. We construct a class of easy-to-implement, distribution-free tests based on new nonparametric tools and unexplored connections with spectral graph theory. The test is shown to possess various desirable properties along with a characteristic exploratory flavor that has practical consequences for statistical modeling. The numerical examples show that our method works surprisingly well under a broad range of realistic situations.
\end{abstract}
\begin{keywords}
Graph-based nonparametrics; High-dimensional k-sample comparison; Spectral graph partitioning; Distribution-free methods; High-dimensional exploratory analysis.
\end{keywords}
\section{Introduction}
\subsection{The Problem Statement}
The goal of this paper is to introduce a new theory for nonparametric multivariate $k$-sample problems. Suppose $k$ mutually independent random samples with sizes $n_1,\ldots n_k$ are drawn, respectively, from the unknown distributions $G_i(x),$ $i=1,\ldots,k$ in $\mathbb{R}^d$. We wish to test the following null hypothesis:
\[H_0: G_1(x) = \cdots = G_k(x) ~~\text{for all}~ x,\]
with the alternative hypothesis being $H_{{\rm A}}: G_r(x) \neq G_s(x)$ for some $r \neq s \in \{1,\ldots,k\}$. This fundamental comparison problem frequently appears in a wide range of data-rich scientific fields: astrophysics, genomics, neuroscience, econometrics, and chemometrics, just to name a few.
\subsection{Existing Graph-based Nonparametric Approaches ($k=2$)} 

Graphs are becoming an increasingly useful tool for nonparametric multivariate data analysis. Their significance has grown considerably in recent years due to their exceptional ability to tackle high-dimensional problems, especially when the dimension of the data is much larger than the sample size--a regime where the classical multivariate rank-based k-sample tests \citep{oja2004,puri1993book} are not even applicable. 
Lionel \cite{weiss1960two} made an early attempt to generalize classical nonparametric two-sample procedures in multivariate settings by utilizing nearest-neighbor-type graphs. Although the idea was intriguing, the main impediment to its use came from the fact that the null distribution of the proposed test statistic was distribution-dependent.  The drawback was later addressed by \cite{friedman1979}, who provided the first practical minimal spanning tree based algorithm, known as the edge-count test. This line of work led to a flurry of research, resulting in some promising new generalizations. For example, \cite{chen2016} noted that the edge-count test has ``low or even no power for scale alternatives when the dimension is moderate to high unless the sample size is astronomical due to the curse-of-dimensionality.'' To counter this undesirable aspect, the authors proposed a generalized edge-count test based on a clever Mahalanobis distance-based solution for two-sample location-scale alternatives. This was further extended in \cite{chen2016weighted} to tackle another shortcoming of Friedman and Rafsky's approach, namely the sample imbalance problem. The authors apply appropriate weights to different components of the classical edge-count statistics to construct a weighted edge-count test. \cite{rosenbaum2005} proposed a test based on minimum distance non-bipartite pairing, and \cite{biswas2014} utilized the concept of shortest Hamiltonian path of a graph for distribution-free two-sample problems.

Broadly speaking, all of the above graph-based methods share the following characteristics:
\vskip.1em
\begin{itemize}[topsep=2pt]
  \setlength{\itemsep}{1.65pt}
\item Define the pooled sample size $n=n_1+\cdots+n_k$. For $i=1,\ldots,n$, let $x_i$ denotes the feature vector corresponding to the $i$-th row of the $n\times d$ data matrix $X$. Construct a weighted undirected graph $\mathcal{G}$ based on pairwise Euclidean distances between the data points $x_i \in \mathbb{R}^d$ for $i=1,\ldots,n$.
\item Compute a subgraph $\mathcal{G}^*,$ containing certain subset of the edges from the original weighted graph $\mathcal{G}$ using some optimality criterion such as shortest Hamiltonian path or minimum non-bipartite matching or minimum weight spanning tree, which often requires sophisticated optimization routines like Kruskal's algorithm or Derig’s shortest augmentation path algorithm for efficient implementation. 
\item Compute cross-match statistics by counting the number of edges between samples from two different populations.
\item The common motivation behind all of these procedures is to generalize the univariate Wald-Wolfowitz runs test \citep{wald1940} for large dimensional problems.
\end{itemize}
\vskip.2em
The general idea here is to reformulate the multivariate comparison problem as a study of the structure of a specially designed graph, one constructed from the given data; see Section \ref{sec:lpgk} for more details. %We end this section with some key references to other related work, like  \cite{schilling1986multivariate,henze1988,hall2002,rousson2002, gretton2012kernel} and \cite{bhattacharya2015}. 
\vspace{-.65em}
\subsection{Some Desiderata and Goals} 
\label{sec:desiderata}
The main challenge is to design a nonparametric comparison test that continues to work for high-dimensional (and low-sample size) multi-sample cases, and is distribution-free.  In particular, an ideal nonparametric multivariate comparison test should 
\vskip.25em
(D1.) be robust, andnot unduly influenced by outliers. This is an important issue, since detecting outliers is quite difficult in large-dimensional settings. Figs \ref{fig:sim1} (e) and (f) show that current methods perform poorly for datasets contaminated by even a small percentage of outliers. %Thus there is a need for robust nonparametric multivariate tests. 
\vskip.25em
(D2.)  allow practitioners to systematically construct tests for high-order alternatives, beyond conventional location-scale. Figs \ref{fig:sim1} (c) and (d) demonstrate the need for such a test.  Complex ``shape'' differences frequently arise in the real-world; see, for example: p53 geneset enrichment data of Section \ref{sec:p53}.
\vskip.25em
(D3.) be valid for any combination of discrete and continuous covariates--another common trait of real multivariate data such as the Kyphosis data in Fig \ref{fig:real1} (b), where existing methods perform poorly. The same holds true for Figs \ref{fig:sim1} (g) and (h), which further reinforce the need to develop mixed-data algorithms.
\vskip.25em
(D4.) provide the user with some insight into why the test was rejected. In fact, data scientists need to be able to see not just the final result as a form of a single p-value, but also to understand and interpret why they are arriving at that particular conclusion, as in tables \ref{table1.1:lpchart}-\ref{table:brain}. These additional insights might lead to enhanced predictive analytics at the next modeling phase, as discussed in Sec. \ref{sec:brain} for brain tumor data with $k=5, d=5597$, and $n=45$.
\vskip.25em
(D5.) work for general $k$-sample problems. Surprisingly, all currently available graph-based comparison tests work only for two-sample problems. The main reason for this is that all of the existing methods aim to generalize the univariate run test to high dimension. Interestingly, the $k$-sample generalization of the run test is known to be hopelessly complicated 
even in the univariate case--so much so that \cite{mood1940}, in discussing the asymptotic distribution of run test for more than two samples, has not hesitated to say ``that such a theorem would hardly be useful to the statistician, and the author does not feel that it would be worthwhile to go through the long and tedious details merely for the sake of completeness.'' A convenient method that works for general $k$-sample comparisons would thus be highly valuable for practical reasons.
\vskip.25em
In this paper, we begin the process of considering an alternative class of nonparametric test based on new tools and ideas, which we hope will be useful in light of D1-D5. 
\section{Theory and Methods}
\label{sec:theory}
\subsection{Nonlinear Data-Transformation}
Suppose we are given $\{(Y_i, X_i):i=1,\ldots, n\}$ where $Y_i \in \{1,\ldots,k\}$ denotes the class membership index, and $X_i \in \mathbb{R}^d$ is the associated multidimensional feature. $n_g$ is the number of samples from class $g$ and $n=\sum_{g=1}^k n_g$. Our first task will be to provide a universal mechanism for constructing a new class of nonparametric data-graph kernel, useful for $k$-sample problems.
\vskip.5em
By $\{F_j\}_{j=1}^d$ we denote the marginal distributions of a d-variate random vector. Given $X_1,\ldots,X_n$ random sample from $F_j$ construct the polynomials $\{T_\ell(X;\wtF_j)\}_{\ell\ge 1}$ for the Hilbert space $\cL^2(\wtF_j)$ by applying Gram-Schmidt orthonormalization on the set of functions $\{\zeta,\zeta^2,\ldots\}$: 
\beq \zeta(x;\wtF_j)~=~\dfrac{\sqrt{12}\big\{\tFm_j(x) - 1/2\big\}}{\sqrt{1-\sum_{x \in \mathscr{U}} \widetilde{p}_j^3(x)}},\eeq
$\widetilde{p}_j$ and $\wtF_j$ denote the empirical probability mass function and distribution function, respectively; $\tFm_j(x)=\wtF_j(x)- .5 \widetilde{p}_j(x)$ is known as the mid-distribution transform; $\mathscr{U}$ denotes the set of all distinct observations of $X$. We call this system of specially-designed orthonormal polynomials of mid-distribution transforms as empirical LP-basis function. Two notable points: First, the number of LP-basis functions $m$ is always less than $|\mathscr{U}|$. As an example, for $X$ binary, we can construct at most one LP-basis. Second, the shapes of these empirical-polynomials are not fixed, they are data-adaptive which make them inherently nonparametric by design. The nomenclature issue of LP-basis is discussed in Supplementary Material S1.
%%%%%%%%%%%%%%%%%%%%%%%%%%%%%%%%%%%%%%%%%%%%%%%%%%%
\subsection{Theoretical Motivation} To appreciate the motivation behind this nonparametrically designed system of orthogonal functions, consider the univariate two-sample ($Y,X$) problem. We start by deriving the explicit formulae of the LP-basis functions for $X$ and $Y$. 

\begin{theorem} \label{thm:1}
Given independently drawn $\{(Y_i,X_i), i=1,\ldots,n=n_1+n_2\}$ where $Y\in \{0,1\}$ and $X \in \mathbb{R}$ with the associated pooled rank denoted by $R_i ={\rm rank}(X_i)$, the first few empirical LP-polynomial bases for $X$ and $Y$ are given by:
\beq \label{eq:yLP}  T_1(y_i;\wtF_Y) ~=~ \left\{ \begin{array}{rl}
 -\sqrt{\dfrac{n_2}{n_1}} &\mbox{for $i=1,\ldots,n_1$} \\
 \sqrt{\dfrac{n_1}{n_2}} &\mbox{for $i=n_1+1,\ldots,n$.~~~~~~~~~~~~~\,}
       \end{array} \right. \eeq
\beq \label{eq:xLP} T_\ell(x_i; \wtF_X) ~=~ \left\{\begin{array}{rl}
 &\dfrac{\sqrt{12}}{n}\Big( R_i - \dfrac{n+1}{2} \Big) ~~~\mbox{for $\ell=1$} \\[.74em]
 &\dfrac{6\sqrt{5}}{n^2} \Big( R_i - \dfrac{n+1}{2} \Big)^2 - \dfrac{\sqrt{5}}{2}~~~ \mbox{for $\ell=2$.~~~~}
       \end{array} \right. \eeq
       where $\ell$ denotes the order of the polynomial.
\end{theorem}
\begin{definition} \label{def:lpc}
Define LP-comeans as the following cross-covariance inner product 
\beq \label{eq:defLPcom} \LP[j,\ell;Y,X]\,=\,\Ex[T_j(Y;F_Y) T_\ell(X;F_X)] ,~~j,\ell>0.\eeq 
\end{definition}
One can easily estimate the LP-comeans by substituting its empirical counterpart: $\hLP[j,\ell;Y,X]=n^{-1}\sum_{i=1}^n T_j(y_i;\wtF_Y) T_\ell(x_i;\wtF_X)$, which can be computed in \texttt{R} by $\texttt{cov}\{T_j(y;\wtF_Y), T_\ell(x;\wtF_X)\}$. This immediately leads to the following surprising identity.
\begin{theorem} \label{thm:wm}
A compact LP-representation of the Wilcoxon and Mood Statistic is given by
\vskip.6em
$~~~~~~~~~~\text{$\hLP[1,1;Y,X]~~\equiv~~$ Wilcoxon statistic for testing equality of location or means},$
\vskip.6em
$~~~~~~~~~~\text{$\hLP[1,2;Y,X]~~\equiv~~$ Mood statistic for testing equality of variance or scale.}~~$
\end{theorem}
% \vspace{-.45em}
% \begin{proof} The proof proceeds by verifying
% \beq \sqrt{n}\hLP[1,1;Y,X] ~= ~\sqrt{\dfrac{12}{n\, n_1 n_2}} \left[ \sum_{i=n_1+1}^{n} R_i\,-\, \dfrac{n_2(n+1)}{2}\right],~~~~~~~~~~~~~\eeq
% and
% \beq \sqrt{n}\hLP[1,2;Y,X] ~= ~\sqrt{\dfrac{180}{n^3\, n_1 n_2}} \sum_{i=n_1+1}^{n}\left[ \left( R_i - \dfrac{n+1}{2} \right)^2\,-\, \dfrac{n^2+2}{12}             \right]. \eeq

The proofs of Theorems \ref{thm:1} and \ref{thm:wm} are deferred to the Appendix. Our LP-Hilbert space inner-product representation automatically produces ties-corrected linear rank statistics when $X$ is discrete, by appropriately standardizing by the factor $1-\sum_{x \in \mathscr{U}} \tp_X^3(x)$; cf. \cite{chanda1963}, and \citet[ch 4 p. 118]{hollander2013book}. This implies the sum of squares of LP-comeans 
\beq \label{eq:lpinfor} \sum_{\ell>0}\big|\LP[1,\ell;Y,X]\big|^2\eeq can provide a credible and unified two-sample test to detect univariate distributional differences. Our goal in this paper will be to introduce a multivariate generalization of this univariate idea that is applicable even when the dimension exceeds the sample size--a highly non-trivial modeling task. The crux of our approach lies in interpreting the data-adaptive LP-basis transformations as specially designed nonlinear discriminator functions, which paves the way for extending it to encompass multivariate problems.
\subsection{LP Graph Kernel} \label{sec:lpgk}
Here we are concerned with LP-graph kernel, the first crucial ingredient for constructing a graph associated with the given high-dimensional data.
\begin{definition} Define the $l$-th order LP Gram Matrix $W_{l}^{\LP}\in \mathbb{R}^{n\times n}$ as
\beq \label{eq:Gram}
W_l^{\LP}(i,j)=\Big(c+ \big\langle \Phi^{\LP}_l (x_i),  \Phi^{\LP}_l (x_j) \big\rangle\Big)^2,
\eeq
a polynomial kernel (degree=$2$) where $\Phi^{\LP}_l:\mathbb{R}^d \rightarrow \mathbb{R}^d$ denotes the feature map in the LP space for $l=1,2,\ldots$
\beq \label{eq:FM}
\Phi^{\LP}_l: (x_1,x_2,\cdots, x_d)^T \mapsto \big\{T_l(x_1,\wtF_1),T_l(x_2,\wtF_2),\cdots,T_l(x_d,\wtF_d)\big\}^T,~~~~
\eeq
as before, the function $T_l(\cdot;\wtF_j)$ denotes $l$-th order empirical LP-polynomial associated with $\wtF_j$.
\end{definition}

The positive symmetric kernel $W_l^{\LP}: \mathbb{R}^d \times \mathbb{R}^d \mapsto \mathbb{R}^+$ encodes the similarity between two $d$-dimensional data points in the LP-transformed domain. From $X$ and $W_l^{\LP}$, one can thus construct a weighted graph $\mathcal{G} = (V, W_l^{\LP})$ of size $n$, where the vertices $V$ are data points $\{x_1,\ldots,x_n\}$ with edge weights given by LP-polynomial kernel $W_l^{\LP}( x_i, x_j)$.
\vskip1em
\textsc{Analysis pipeline}:~$X\in \mathbb{R}^{n\times d}\longrightarrow \text{$\ell$-th order LP-Transform}\longrightarrow W_l^{\LP}\in \mathbb{R}^{n\times n}\longrightarrow \mathcal{G}(V, W_l^{\LP}).$
\vskip1em
The resulting graph captures the topology of the high-dimensional data cloud in the LP-transformed domain. In the next section, we introduce a novel reformulation of the $k$-sample problem as a supervised structure learning problem of the learned LP-graph.

\subsection{Equivalence to Graph Partitioning} 
\label{sec:ncut}
Having reformulated the high-dimensional comparison as a graph problem, we can focus on understanding its structure. In the following, we describe an example to highlight how LP graph-based representation can provide rich insights into our $k$-sample learning objectives. 

\begin{figure}[t]
\centering
\includegraphics[width=0.38\textwidth,keepaspectratio,trim=2.65cm .25cm 2.65cm .25cm]{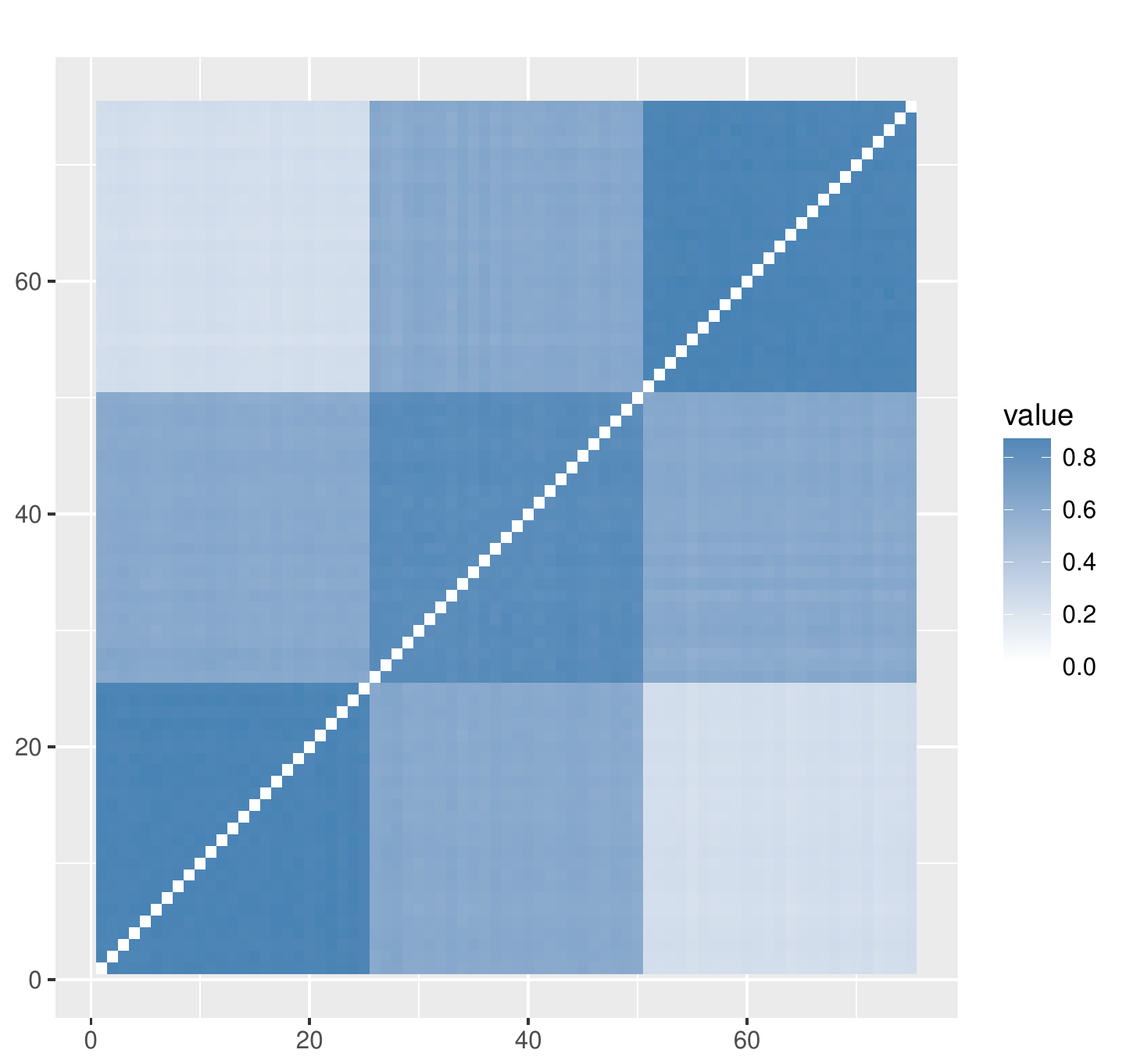}
\vskip.55em
\caption{Heatplot of $W_1^{\LP}$ Map for the $3$-sample location-alternative model discussed in Sec \ref{sec:ncut}. The samples are organized according to groups.}
\label{fig:hmap}
\vspace{-1.2em}
\end{figure}

Consider a three-sample testing problem based on the following location-alternative model: $G_i=\cN_d(\delta_i 1_d, I_d)$ with $\delta_1=0, \delta_2=1.5, \delta_3=3$, dimension $d=500$, and the sample sizes $n_i$'s are equal to $25$ for $i=1,2$ and $3$. Fig \ref{fig:hmap} shows the inhomogeneous connections of the learned graph $\mathcal{G} = (V, W_1^{\LP})$. It is evident from the connectivity matrix that LP-edge density $W_1^{\LP}(i,j)$ is significantly ``higher'' when both $x_i$ and $x_j$ come from the same distribution than when they arise from different distributions. This creates a natural clustering or division of the vertices into $k$-groups (here $k=3$) or communities. Naturally, under the null hypothesis (when all $G_i$'s are equal), we would expect to see one homogeneous (edge densities) graph of size $n$ with no community structure. In order to formalize this intuition of finding densely connected subgraphs, we need some concepts and terminology first.

The objective is to partition $V$ into $k$ non-overlapping groups of vertices $V_g, ~g=1,\ldots,k$, where in our case $k$ is \textit{known}. To better understand the principle behind graph partitioning, let us first consider the task of grouping the nodes into two clusters (i.e., the two-sample case). A natural criterion would be to partition the graph into two sets $V_1$ and $V_2$ such that the weight of edges connecting vertices in $V_1$  to vertices in $V_2$ is minimum. This can be formalized using the notion of graph cut:
\[\mbox{Cut}(V_1,V_2)=\sum_{i\in V_1, j\in V_2} W^{\LP}(i,j).\]
By minimizing this cut value, one can optimally bi-partition the graph. However, in practice, minimum cut criteria does not yield satisfactory partitions and often produces isolated vertices due to the small values achieved by partitioning such nodes.

One way of getting around this problem is to design a cost function that prevents this pathological case. This can be achieved by normalizing the cuts by the volume of $V_i$, where ${\rm Vol}(V_i)=\sum_{j \in V_i} \tdeg_j$ and $\tdeg_i$ is the degree of $i$-th node defined to be $\sum_{j=1}^n W_l^{\LP}(i,j)$. Partitioning based on normalized cut (Ncut) is a means of implementing this idea, since doing so minimizes the following cost function:
\beq \label{eq:ncut1} \mbox{NCut}(V_1,\ldots,V_k)= \sum_{g=1}^k \dfrac{\mbox{Cut}(V_g,V- V_g)}{{\rm Vol}(V_g)}, \eeq
which was proposed and investigated by \cite{shi2000}.
\begin{theorem} [Chung, 1997]
Define the indicator matrix $\Psi=(\psi_1,\ldots,\psi_k)$ where
\beq\label{eq:ncut2} \psi_{j,g} ~=~ \left\{ \begin{array}{rl}
 & \sqrt{\frac{\tdeg_j}{{\rm Vol}(V_g)}},~~ \mbox{if $j \in V_g$} \\
& ~~0,~~\textrm{otherwise.~~}
       \end{array} \right. \eeq
for $j=1,\ldots,n$ and $g=1,\ldots,k$. Then the k-way Ncut minimization problem \eqref{eq:ncut1} can be equivalently rewritten as
\beq \label{eq:ncut3}
\text{$\min_{V_1,\ldots,V_k} {\rm Tr}\big( \Psi^T \cL \Psi\big)$~ subject to ~$\Psi^T\Psi=I$, and $\Psi$ as in \eqref{eq:ncut2}},
\eeq
where $\cL$ is known as the (normalized) Laplacian matrix given by $D^{-1/2}W^{\LP}D^{-1/2}$ and $D$ is the diagonal matrix of
vertex degrees. \nocite{chung1997}
\end{theorem}
This discrete optimization problem \eqref{eq:ncut3} is unfortunately NP-hard to solve. Hence, in practice relaxations are used for finding the optimal Ncut by allowing the entries of the $\Psi$ matrix to take arbitrary real values: 
\beq \label{eq:ncutRelax} \text{$\min_{\Psi \in \mathbb{R}^{n\times k}} {\rm Tr}\big( \Psi^T \cL \Psi\big)$ ~subject to~ $\Psi^T\Psi=I$}.\eeq
The resulting trace-minimization problem can be easily solved by choosing $\Psi$ to be the first $k$ generalized eigenvectors of $\cL u=\la D u$ as columns--Rayleigh-Ritz theorem. Spectral clustering methods \cite[Sec. 5]{von2007} convert the real-valued solution into a discrete partition (indicator vector) by applying k-means algorithms to the rows of  eigenvector matrix $U$; see Remark R1 in the Appendix.
%%%%%%%%%%%%%%%%%%%%%%%%%%%%%%%%%%%%%%%%%%%%%
\begin{figure}[h]
\vskip.65em
\includegraphics[width=\linewidth,keepaspectratio,trim=.5cm .4cm .5cm 0cm]{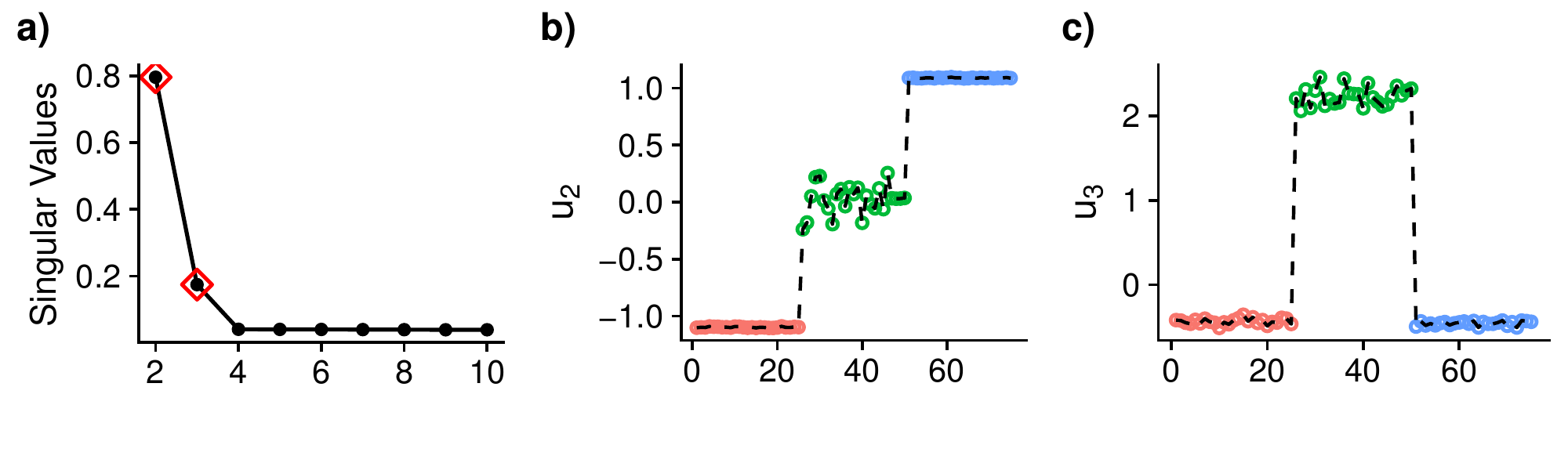}
\caption{Laplacian spectral analysis of 3-sample location problem using LP-graph kernel (cf. Sec \ref{sec:ncut}). The non-trivial singular values are shown in (a). The top ones $(\la_2,\la_3)$ are marked with red diamonds. (b) and (c) display the dominant singular vectors $(u_2,u_3)$; colors indicate the true group labels ($Y$).}
\label{fig:3sampleLapspec}
\end{figure}
%%%%%%%%%%%%%%%%%%%%%%%%%%%%%%%%%%%%%%%%%%%%%

Fig \ref{fig:3sampleLapspec} displays the spectrum of the Laplacian matrix for the 3-sample location-alternative model based on LP-graph kernel. As is evident from the figure, the top two dominant singular vectors accurately identify the original group's structure encoded in $Y$. In fact, \cite{von2008consistency} proved that as $\nti$ the Laplacian spectral clustering converges to the true cluster labels under very mild conditions which are usually satisfied in real world applications. This classical result has recently been extended \citep{karoui2010,couillet2016kernel} to the big-data regime where the dimension $d$ and sample size $n$ grow simultaneously, by applying the spiked random matrix theory and concentration of measure results.
%%%%%%%%%%%%%%%%%%%%%%%%%%%%%%%%%%%%%%%%%%%%%
\subsection{Test Statistics and Asymptotic Properties} \label{sec:tstat}
Following the logic of the previous section, we identify the hidden $k$ communities simply by clustering (using k-means algorithm) each row of $U$ as a point in $\mathbb{R}^{k}$. We store the cluster assignments in the vector $Z$, where $Z_i\in \{1,\ldots,k\}$. It is important to note that up until now we have not used the true labels or the group information $Y_i$ for each data point, where $Y_i \in \{1,\ldots,k\}$ for $i=1,\ldots,n$.

At each node of the graph, we now have the bivariate $(Y_i,Z_i)$, which can be viewed as a map $V \mapsto \{1,\ldots, k\}^2$. This data-on-graph viewpoint will allow us to convert the original high-dimensional $k$-sample testing problem into a more convenient form. At this point, an astute reader might have come up with an algorithm by recognizing that the validity of the $k$-sample null-hypothesis can be judged based on how closely $Y$ the group index variable is correlated with the intrinsic community structure $Z$ across the vertices. This turns out to be an absolutely legitimate algorithm. In fact, when the null hypothesis is true, one would expect that the $Z_i$'s can take values between $1$ and $k$ almost randomly, i.e. $\Pr(Z_i=g)=1/k$ for $g=1,\ldots,k$ and $i \in V$. Thus the hypothesis of equality of $k$ high-dimensional distribution can now be reformulated as an independence learning problem over graph $H_0$: $\texttt{Independence}(Y, Z)$. Under the alternative, we expect to see a higher degree of dependence between $Y$ and $Z$. 

The fundamental function for dependence learning is the `normed joint density,' pioneered by \cite{Hoeff40}, defined as the joint density divided by the product of the marginal densities:
\[\texttt{dep}(y_i,z_i;Y,Z)=\dfrac{p(y_i,z_i;Y,Z)}{p(y_i;Y) p(z_i;Z)},\]
which is a `flat' function over the grid $\{1,\ldots,k\}^2$ under independence.
\begin{definition}
For given discrete $(Y,Z)$, the bivariate copula density kernel $\cd: [0,1]^2 \rightarrow \mathbb{R}_+ \cup \{0\}$
is defined almost everywhere through
\beq \label{eq:copdef}
\cd(u,v;Y,Z)\,=\,\texttt{dep}\big\{Q(u;Y),Q(v;Z);Y,Z\big\}\,=\,\dfrac{p\big\{ Q(u;Y), Q(v;Z);Y,Z\big\}}{p\big\{ Q(u;Y) \big\} p\big\{  Q(v;Z) \big\}},~~~0<u,v<1
\eeq
\end{definition}
where $Q(\cdot)$ denotes the quantile function. It is not difficult to show that this quantile-domain copula density is a positive piecewise-constant kernel satisfying
\[\iint_{[0,1]^2} \cd(u,v;Y,Z) \dd u \dd v\,~ = \sum_{(i,j)\in \{1,\ldots,k\}^2} \iint_{I_{ij}} \cd(u,v;Y,Z) \dd u \dd v ~= ~1,\]
where
\[I_{ij}(u,v)= \left\{ \begin{array}{ll}
         \,1, ~~& {\rm if}~ (u,v) \in \left(F_Y(i), F_Y(i+1)\right] \times \left( F_Z(j), F_Z(j+1)\right]\\
         \,0, ~~& \mbox{elsewhere}.\end{array} \right.\]
\begin{theorem} \label{th4}
The discrete checkerboard copula $\cd(u,v;Y,Z)$ satisfies the following important identity in terms of LP comeans between $Y$and $Z$:
\[\iint_{[0,1]^2} (\cd-1)^2\dd u \dd v ~= ~\sum_{j=1}^{k-1}\sum_{\ell=1}^{k-1}  \big|\LP[j,\ell;Y,Z]\big|^2.\]
\end{theorem}
\begin{proof} The key is to recognize that
\[\iint\limits_{u,v} \cd(u,v;Y,Z) T_j\{Q(u;Y);F_Y\} T_\ell\{Q(v;Z);F_Z\} \dd u \dd v=\mathop{\sum\sum}_{y,z}  T_j(y;F_Y) T_\ell(z;F_Z) p(y,z;Y,Z),\]
which, using definition 1, can be further simplified as
\[\Ex[T_j(Y;F_Y) T_\ell(Z;F_Z)]~=~\LP[j,\ell;Y,Z].\] 
Apply Parseval's identity on the LP-Fourier copula density expansion to finish the proof. 
\end{proof}
Theorem \ref{th4} implies that a test of independence is equivalent to determining whether all the $(k-1)^2$ parameters $\LP[j,\ell;Y,Z]$ are zero. Motivated by this, we propose our \underline{G}raph-based \underline{LP}-nonparametric (abbreviated to GLP) test statistic: 
\beq \label{eq:glp} {\rm GLP}~ \text{statistic}~=~ \sum_{j=1}^{k-1}\sum_{\ell=1}^{k-1}  \big|\hLP[j,\ell;Y,Z]\big|^2.
\eeq
The test statistics \eqref{eq:glp} bear a surprising resemblance to \eqref{eq:lpinfor}, thus fulfilling a goal that we had set out to achieve. 
\begin{theorem} \label{them:an}
Under the independence, the empirical LP-comeans $\hLP[j,\ell;Y,Z]$ have the following limiting null distribution as $\nti$
\[\sqrt{n} \hLP[j,\ell;Y,Z] \,\overset{i.i.d}{\sim} \,\cN(0,1).\]
\end{theorem}
Appendix \ref{appendix:proof5} contains the proof. Theorem \ref{them:an} readily implies that the asymptotic null distribution of our test statistic \eqref{eq:glp} is chi-square distribution with $(k-1)^2$ degrees of freedom. However, it is instructive to investigate the accuracy of this chi-square null distribution in approximating p-values for finite samples. The boxplots in Figure \ref{fig:nullasymp_comp} display the differences between the asymptotic p-values and permutation p-values for $G_1=G_2=\cN_d(0,I_d)$. 
The permutation p-values were computed by randomly shuffling the class-labels $Y$ $1000$ times. Each simulation setting was repeated $250$ times for all the combinations of $n_1, n_2$ and $d$. From the boxplots, we can infer that the theoretically predicted asymptotic p-values provide an excellent approximation to the permutation p-values for sample size $n \ge 100$. Furthermore, the limiting chi-square based p-values approximation continues to hold even for increasing dimension $d$. 
\begin{figure}[t]
\centering
\includegraphics[width=.92\linewidth,keepaspectratio
,trim=1cm 1.4cm 1cm 1cm]{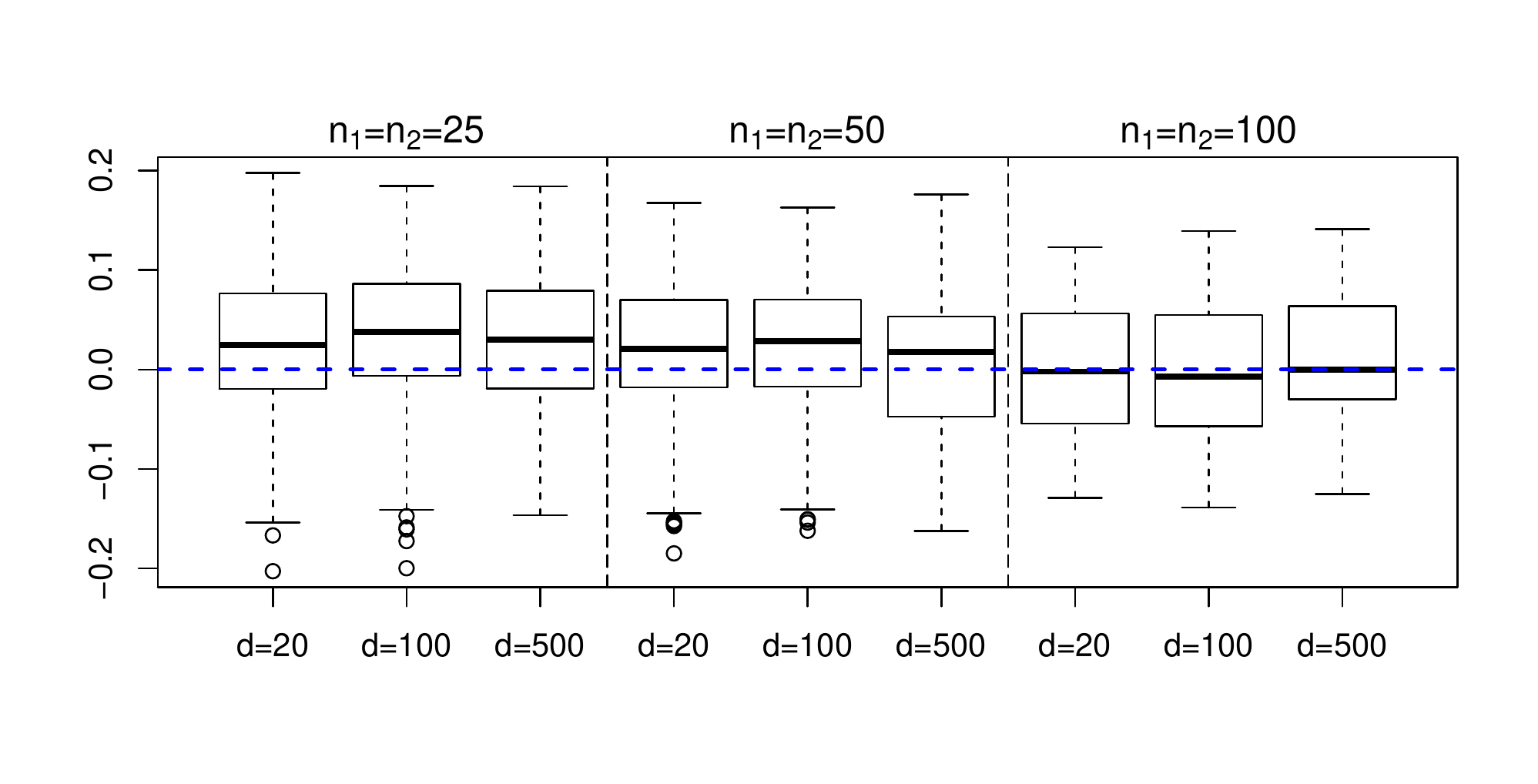}~~
\caption{Boxplots of the difference between asymptotic p-values and p-values from 1000 permutations, compared across different sample sizes $n_1$, $n_2$ and dimension $d$. Each setting was repeated $250$ times.}
\label{fig:nullasymp_comp}
\vspace{-1.65em}
\end{figure}
\subsection{Algorithm}
\label{sec:algo}
This section describes the steps for our $k$-sample testing procedure that combines the strength of modern nonparametric statistics with spectral graph theory to produce an interpretable and adaptable algorithm. 
\begin{center}

{\bf \texttt{GLP}: Graph-based Nonparametric k-sample Learning}
\end{center}
\vspace{-.8em}
\medskip\hrule height .7pt
\vskip.65em
~~~\,\texttt{Step 1.} Input: We observe the group identifier vector $Y$ and the data matrix $X \in \mathbb{R}^{n\times d}$. The features could be discrete, continuous, or even mixed. 
\vskip.4em
\texttt{Step 2.} Construct LP-graph kernel $W_{\ell}^{\LP}\in \mathbb{R}^{n\times n}$ using \eqref{eq:Gram}. The choice of $\ell$ depends on the type of testing problem: $\ell=1$ gives $k$-sample test for mean, $\ell=2$ for scale alternative and so on.  
\vskip.4em
\texttt{Step 3.} Compute normalized Laplacian matrix $\cL$ for the LP-learned graph $\mathcal{G} = (V, W_{\ell}^{\LP})$ by $D^{-1/2}W_{\ell}^{\LP}D^{-1/2}$, where $D$ is the the diagonal degree matrix with elements $W_{\ell}^{\LP}1_n$.
\vskip.4em
\texttt{Step 4.} Perform spectral decomposition of $\cL$ and store leading nontrivial $k-1$ eigenvectors in the matrix $U \in \mathbb{R}^{n\times k-1}$. 
\vskip.4em
\texttt{Step 5.} Apply k-means clustering by treating each row of $U$ as a point in $\mathbb{R}^{k-1}$ for Ncut community detection \eqref{eq:ncutRelax}. Let $Z$, a vector of length $n$, denotes the cluster assignments obtained from the k-means.
\vskip.4em
\texttt{Step 6.} For each node of the constructed graph we now have bivariate random sample $(Y_i,Z_i)$ for $i\in V$. Perform correlation learning over graph by computing the GLP statistic \eqref{eq:glp}. Compute the p-value using $\chi^2_{(k-1)^2}$ null distribution, as described in Sec \ref{sec:tstat}.
\vskip.4em
\texttt{Step 7.} For complex multi-directional testing problems, fuse the LP-graph kernels to create a super-kernel by $W^{\LP}=\sum_\ell W_{\ell}^{\LP}$.
\vskip.25em
~~~~(\texttt{7a})~ Merging: As an example, consider the general location-scale alternatives--a targeted testing. Here we compute $W^{\LP}$ by taking sum of $W_{1}^{\LP}+W_{2}^{\LP}$ and repeating steps 3-6 for the testing. 
\vskip.4em
~~~\,(\texttt{7b})~ Filtering: Investigators with no precise knowledge of the possible alternatives can combine informative $W_{\ell}^{\LP}$ based on the p-value calculation of step 6 after adjustment for multiple comparisons to construct a tailored graph kernel, as in tables \ref{table1.1:lpchart}, \ref{table3:leuk} and \ref{table:brain}.
\medskip\hrule height .65pt
\vskip1em
\subsection{Illustration Using p53 Gene-set Data} \label{sec:p53}
Here we demonstrate the functionality of the proposed GLP algorithm using p53 geneset data (\texttt{www.broad.mit.edu/gsea}): it contains transcriptional profiles of $10,100$ genes in $50$ cancer cell lines over two classes: $n_1=17$ classified as normal and $n_2=33$ as carrying mutations. The genes were cataloged by \cite{subramanian2005gene} into $522$ genesets based on known biological pathway information.
Naturally, the goal is to identify the multivariate oncogenesets that are differentially expressed in case and control samples. 

For illustration purposes, we focus on two genesets: (a) ``SA\_G1\_AND\_S\_PHASES" with number  of genes $d=14$, and (b) ``anthraxPathway" with number of genes $d=2$. Without any prior knowledge about the type of alternatives, we apply steps 2-6 of  GLP algorithm to create the component-wise decomposition, as depicted in Table \ref{table1.1:lpchart}. Our method decomposes the overall statistics, quantifying the departure from the null hypothesis of equality of high-dimensional distributions, into different orthogonal components. By this means, it provides insights into the possible alternative direction(s). For example, the pathway ``SA\_G1\_AND\_S\_PHASES" shows a location shift, whereas the output for ``anthraxPathway" indicates the difference in the tails. This refined insight could be very useful to a biologist who seeks to understand more deeply \textit{why} any particular geneset is important. Once we identify the informative directions, we compute the super-kernel (step 7  of our algorithm) by fusing $W^{\LP}=\sum_{\in \,\textrm{sig.}\, \ell} W_\ell^{\LP}$.
In the case of ``anthraxPathway," this simply implies $W^{\LP} \leftarrow \,W_4^{\LP}$ and so on for other genesets. This combined $W^{\LP}$ effectively packs all the principal directions (filtering out the uninteresting ones) into a single LP-graph kernel. At the final stage, we execute steps 3-5 using this combined kernel to generate the overall $k$-sample GLP statistic along with its p-value, shown in the bottom row of Table \ref{table1.1:lpchart}. 
\begin{table}[!hht]
\caption{GLP multivariate $k$-sample test for p53 data. The table shows the output of our GLP algorithm, given in Sec \ref{sec:algo}. The overall statistic provides the global $k$-sample confirmatory test, while the individual `components' give exploratory insights into how the multivariate distributions are different. The significant components based on p-values adjusted for multiple comparisons are marked with an asterisk `*'}
\centering
\begin{tabularx}{.8\linewidth}{f ffff}
\toprule
\multirow{2}{*}{Component} & \multicolumn{2}{c}{(a)} & \multicolumn{2}{c}{(b)}\\
\cmidrule(lr){2-3}\cmidrule(lr){4-5}
 & GLP & p-value &  GLP & p-value\\
\midrule
1 & 0.145 & 0.007* & $3.84\times 10^{-4}$ & 0.890\\
2 & 0.045 & 0.136  & 0.003 & 0.720\\
3 & 0.002 & 0.754 & 0.045 & 0.136\\
4 & 0.034 & 0.196  & 0.131 & 0.011*\\
\midrule
overall& 0.145 & 0.007 &0.125 &0.012\\
\bottomrule
\end{tabularx}
\label{table1.1:lpchart}
\vskip.25em
\end{table}

Interestingly, \cite{subramanian2005gene} found the geneset ``SA\_G1\_AND\_S\_PHASES" to be related to P53 function but missed ``anthraxPathway,'' the reason being they have used Kolmogorov-Smirnov test which is known to exhibit poor sensitivity for tail deviations; also see Supplementary Material S10.
\begin{table}[h]
\setlength{\tabcolsep}{12pt}
\tbl{Comparing the five tests for two genesets: (a) SA\_G1\_AND\_S\_PHASES;  (b) anthraxPathway. The dimensions are $14$ and $2$, respectively. The reported numbers are all p-values}{
\begin{tabularx}{.85\linewidth}{Y Y Y Y Y Y}
\toprule
Geneset & GLP & FR & Rosenbaum & HP & GEC\\
\midrule
(a) & 0.007 & 0.025 & 0.061 & 0.041 & 0.141\\
(b) & 0.010 & 0.327 & 0.591 & 0.078 & 0.904\\
\bottomrule
\end{tabularx}}
\vskip.3em
\label{table1.3:p53comp}
\begin{tabnote} Methods abbreviation: FR is Friedman \& Rafsky's test; GEC is generalized edge count method by \cite{chen2016}; HP is Biswas' shortest Hamiltonian path method. And GLP is the proposed nonparametric graph-based test.
\end{tabnote}
\end{table}

Finally, it is also instructive to examine our findings in light of other existing methods. This is done in Table \ref{table1.3:p53comp}. In the first case (a), where we have location difference, all methods have no problem finding the significance. However, interestingly enough, most of the competing methods start to fail for (b) where we have higher-order tail difference. This will be further clarified in the next section, where we present numerical studies to better understand the strengths and weaknesses of different methods.

% \includegraphics[height=.7\textheight,keepaspectratio,trim=8cm 0cm 7cm 0cm]{sig_tab.png}
% \vskip.25em
% \caption{Exploratory graphical plot. Top $50$ gene sets (sorted from bottom to top) of p53 data is shown. The color show the magnitude of different components of GLP statistics for a specific geneset.}
% \label{fig7:p53sig}
% \end{figure}
\section{Numerical Results}
\subsection{Power Comparison}
We conduct extensive simulation studies to compare the performance of our method with that of four other methods: Friedman-Rafsky's edge-count test, generalized edge-count test on 5-minimum spanning tree, Rosenbaum's cross matching test, and Shortest Hamiltonian Path method. All the simulations in this section are performed as follows: (i) We focus on the two-sample case, as all the other methods are only applicable in this setup ($k=5$ case is reported in section \ref{sec:brain}) with sample sizes $n_1=n_2=100$; (ii) each case is simulated $100$ times to approximate the power of the respective test. The performance of the algorithms was examined under numerous realistic conditions, as we will shortly note.
\begin{figure}[t]
\centering
\includegraphics[height=.38\textheight, width=0.95\textwidth,trim=1cm .8cm 1cm 0cm]{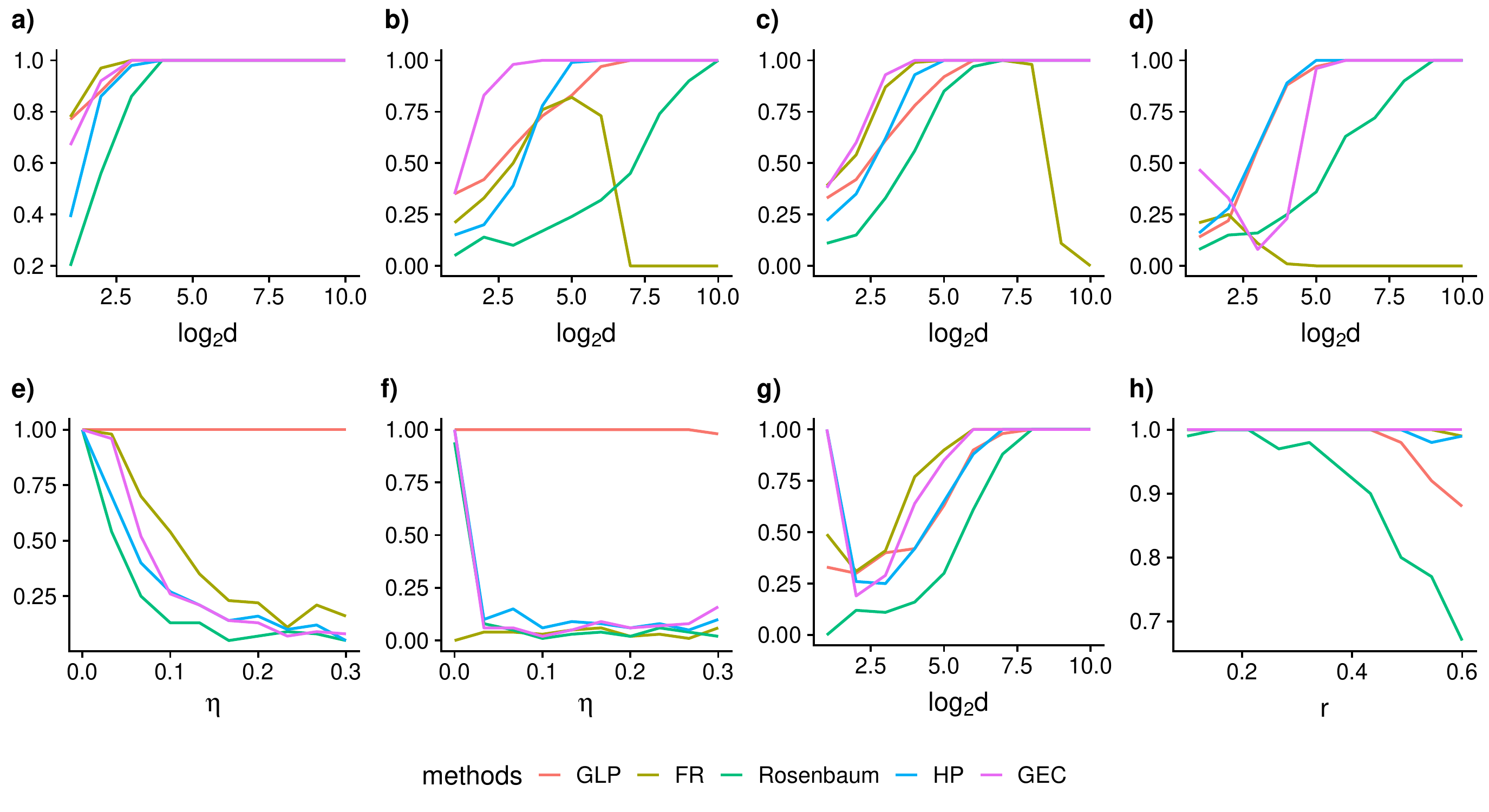} \vskip1em
\vskip1.5em
\caption{(color online) Power comparisons: (a) Location difference $\cN(0,I)$ vs $ \cN(0.5, I)$; (b) Scale difference $\cN(0,I)$ vs $ \cN(0, 1.5I)$; (c) $N(0,I)$ vs $N(0.3 1_d, 1.3I)$; (d) Tail detection $N(0,I_d)$ vs $\mathcal{T}_3(0,I_d)$; (e) $\eta$-Contaminated location difference, with $\eta$ as the percentage of outliers; (f) Tail detection in the presence of outliers; (g) Discrete case ${\rm Poisson}(5)$ vs ${\rm Poisson}(5.5)$; and (h) Mixed alternatives. For (e), (f) we have used $d=500$, for (h), $d=100$, and for constructing the LP-graph kernel \eqref{eq:Gram} we have used $c=0.5$.}
\label{fig:sim1}
\end{figure}

In the first example, we investigate the location alternative case with two groups generated by $G_1=\cN(0,I_d)$ and $G_2= \cN(0.5\mathbf{1},I_d)$ with dimension $d$ ranging from $2$ to $2^{10}=1024$. The result is shown in Fig. \ref{fig:sim1} (a). For small and medium dimensions, our proposed test performs best;  for moderately large dimensions all the methods are equally powerful. In example two, we examine the scale case by choosing $G_1= \cN(0,I_d)$ and $G_2=\cN(0,1.5I_d)$. Here generalized edge-count test reaches the best performance, followed by our test; Friedman and Rafsky’s test completely breaks down. The third example is the general location-scale case $G_1=\cN(0,I_d)$ and $G_2=\cN(0.3 1_d,1.3I_d)$. The estimated power function is shown in Fig \ref{fig:sim1} (c). Our method still displays high performance, followed by the Hamiltonian Path. Example four explores the tail alternative case: $G_1=\cN(0,I_d)$ and $G_2$ is $\mathcal{T}_3(0,I_d)$ Student's t-distribution with degrees of freedom $3$. Not surprisingly, edge-count and generalized edge-count tests suffer the most, as they are not designed for these sorts of higher-order complex alternatives. Both our approach and the Hamiltonian test exhibit excellent performance, which also explains our Table \ref{table1.3:p53comp} finding for  ``anthraxPathway." Discrete data is also an important aspect of two-sample problems. We checked the case of location difference by generating the samples with $G_1={\rm Poisson}(5)$ vs $G_2={\rm Poisson}(5.5)$ in the fifth example, which is depicted in Fig \ref{fig:sim1} (g). Here all methods perform equally well, especially for large-dimensional cases. Example six delves into the important robustness issue. For that purpose, we introduce perturbation of the form $\epsilon_i \sim (1-\eta)N(0,1) +\eta N(\delta_i, 3) $ where $\delta_i=\pm 20$ with probability $1/2$ and $\eta$ varies from $0$ to $0.3$. Empirical power curves are shown in Figs \ref{fig:sim1} (e) and (f). Our proposed test shows extraordinary stability in both cases. The rest of the algorithms reveal their extreme sensitivity towards outliers, so much so that the presence of even $.05\%$ in location-alternative case can reduce the efficiency of the methods by nearly $80\%$, which is a striking number. In the final example we explore the interesting case of mixed alternatives. Here the idea is to understand the performance of the methods when different kinds of alternative hypotheses are mixed up. To investigate that we generate the first group from $G_1=\cN(\mathbf{0},I_d)$, and for the alternative group we generate a portion (also 50 samples) from $\cN(0.3\mathbf{1},I_{d_1})$ and another portion from $\cN(0,1.3I_{d_2})$, where $d=d_1+d_2$, and $r=d_2/d$.  Fig \ref{fig:sim1} (h) shows that generalized edge-count test, Hamiltonian method, and our test perform best. Additional simulation results are given in the Supplementary Materials S2-S6.

\begin{figure}[t]
\centering
~~~~~\includegraphics[width=0.6\textwidth,keepaspectratio,trim=2.5cm 0cm 2.5cm 0cm]{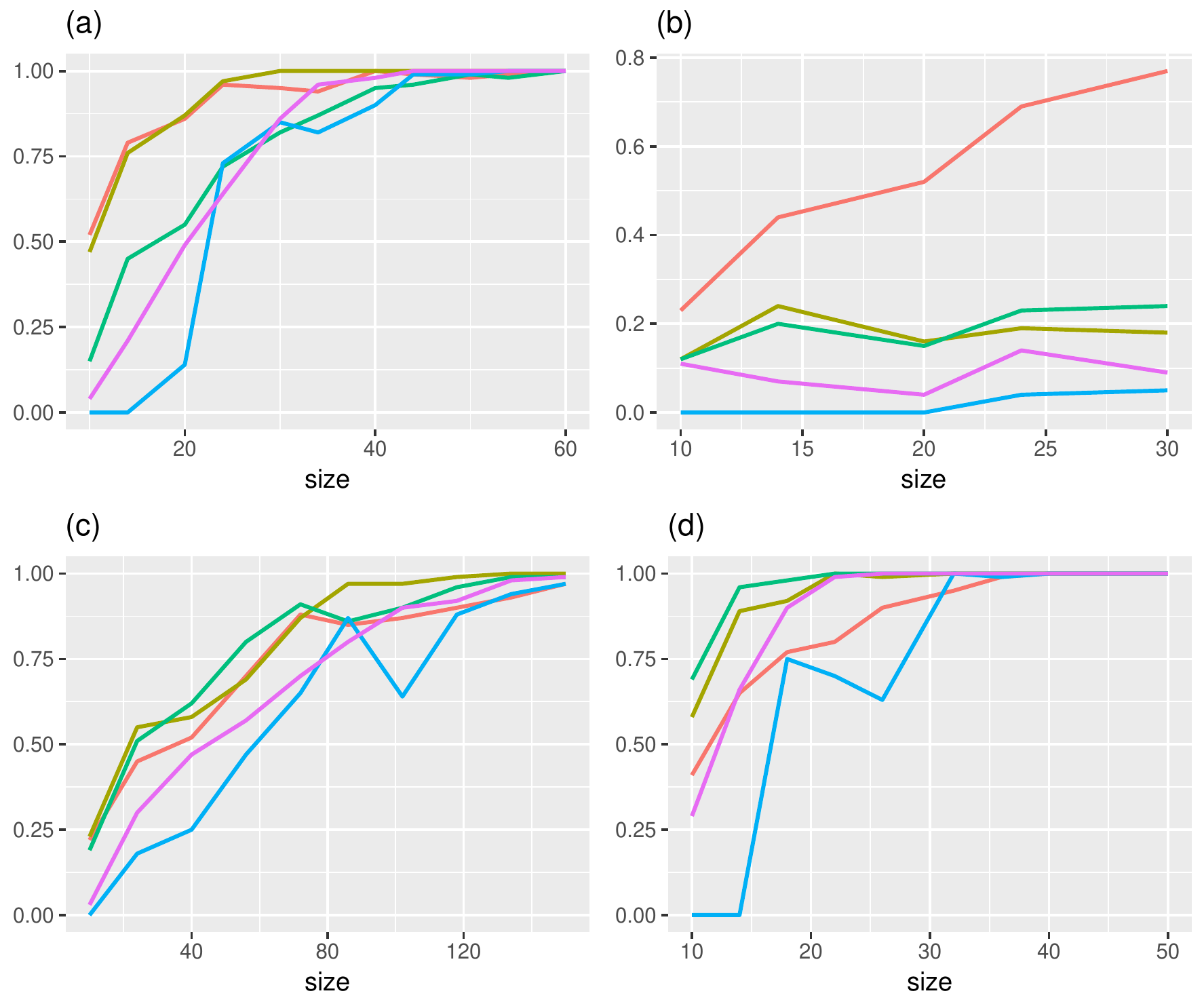}
\includegraphics[width=3in, height=0.3in]{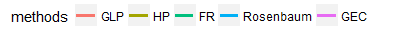}
\caption{Real data empirical power comparisons: (a) Ionosphere data; (b) Kyphosis; (c) Phoneme Data; and (d) Leukemia Data.}
\label{fig:real1}
\end{figure}

\subsection{Analysis of Benchmark Datasets}
Comparisons using the above methods are also done on benchmark datasets. For each dataset, we performed the testing procedures on subsets of the whole data, so that we can approximate the rejection power. For our resampling, (i) several sub-sample sizes are specified, and we pick randomly and evenly from the two groups in full data to form the subsets (ii) Each resampling is repeated 100 times. The results are shown on Figure \ref{fig:real1}.

The first example studied Ionosphere data found on the UCI machine learning repository. The dataset is comprised of $d=34$ features and $n=351$ instances grouped into two classes: good or bad radar returns with $n_1=225$, $n_2=126$. Re-samplings are performed with subsample size ranging from $10$ to $60$. Figure \ref{fig:real1} (a) shows the proposed method performs on par with the Hamiltonian path method, while others fall behind in power.
The next example is kyphosis laminectomy data, available in \texttt{gam} R package. It has discrete covariates whose range of values differ greatly. We observe $17$ instances from the kyphosis group and $64$ instances from the normal group.  Figure \ref{fig:real1} (b) shows that all the existing tests yield noticeably poor performance. Even more surprisingly, the power of these tests does not increase with sample size. Our proposed test significantly outperforms all other competitors here. Next, we consider the phoneme dataset with  two groups `aa' and `ao.' The data have a dimension of $d=256$, and re-sampling subsample sizes range from $10$ to $150$. Here edge-count and Hamiltonian path methods show better performance, and our method also performs well, as shown in Figure \ref{fig:real1} (c). Our final example is the Leukemia cancer gene expression data with $d=7128$ genes in $n_1=47$ ALL (Acute lymphoblastic leukemia) samples and $n_2=25$ AML (Acute myeloid leukemia) samples. Data are re-sampled using total sample sizes from $10$ to $50$. In the case of small subset size the competing methods methods show higher power, with the exception of Rosenbaum's test. Nevertheless, for moderately large sample sizes all methods show equal power. The excellent performance of Friedman and Rafsky's method for this dataset can be easily understood from GLP table \ref{table3:leuk}, which finds only the first component to be significant, i.e., the location-only alternative.
\subsection{K-sample: Brain Data and Enhanced Predictive Model} \label{sec:brain}
The brain data \citep{pomeroy02} contains $n = 42$ samples from $d=5597$ gene expression profiles spanning $k=5$ different tumor classes of the central nervous system with group sizes: $n_1=n_2=n_3=10$, $n_4=8$ and $n_5=4$. This dataset is available at http://www.broadinstitute.org/mpr/CNS. We use this dataset to illustrate our method's performance in testing $k$-sample problems where other methods are not applicable. 

To start with, we have no prior knowledge about the possible alternatives. Thus we first compute the component-wise p-values using our $k$-sample learning algorithm described in Sec \ref{sec:algo}. The result is shown in Table \ref{table:brain}, which 
finds GLP orders $1$ and $2$ as the informative directions. Following step 7b of our algorithm, we then fuse the significant LP-graph kernels to construct the super-kernel $W^{\LP}=\sum_{\ell=1}^2 W_{\ell}^{\LP}$. Applying our spectral graph correlation algorithm on this specially tailored kernel $W^{\LP}$ yields the final row of Table \ref{table:brain}. Combining all of these, we infer that the high-dimensional gene-expression distributions differ in locations and scales in five different tumor classes. 

\begin{minipage}[t][][t]{0.48\textwidth}
\centering
\captionof{table}{Leukemia GLP Chart}
\begin{tabularx}{.95\linewidth}{YYY}
\toprule
Component & GLP &  p-value\\
\midrule
1* & 0.209 & $1.04\times10^{-4}$\\
2 & 0.022 & 0.207\\
3 & 0.002 & 0.703\\
4 & 0.033 & 0.121\\
\midrule
Overall & $0.209$& $ 1.04\times10^{-4}$ \\ %%0.0001
\bottomrule
\end{tabularx}~~~~~
\label{table3:leuk}
\end{minipage}%c=0.5
\begin{minipage}[t][][t]{0.48\textwidth}
\centering
\captionof{table}{GLP chart for Brain Data}
\begin{tabularx}{\linewidth}{YYY}
\toprule
Component & GLP & p-value\\
\midrule
1* & 0.235 & 0.021\\
2* & 1.933 & $5.83\times10^{-30}$\\
3 & 0.070 & 0.760\\
4 & 0.070 & 0.760\\
\midrule
Overall&$0.233$ &$ 0.022$\\
\bottomrule
\end{tabularx}
\label{table:brain}
\end{minipage}
\vskip.55em
Consequently, our technique provides the global confirmatory testing result along with the insights into the possible reasons for rejecting the null hypothesis. This additional information, which is absent from extant technologies, allows us to go further than testing by designing enhanced predictive algorithms. To investigate this assertion, we designed two learning schemes based on two different feature matrices: (i) the original $X \in \mathbb{R}^{n\times d}$, and the data-adaptive (ii) $\mathcal{T}\in \mathbb{R}^{n\times 2d}$ which is $[\mathbb{T}_1\mid \mathbb{T}_2]$, where the $j$th column of the matrix $\mathbb{T}_\ell$ is simply $T_\ell(x;\wtF_j)$--the $\ell$-th LP-transform of covariate $X_j$. We use $75\%$ of the original data as a training set and the rest as testing set; we repeat this process $250$ times. The models are fitted by multinomial lasso-logistic regression; leave-one-out cross validation is used to select the tuning parameter. For comparison we use multi-class log-loss error given by $-\frac{1}{n_{\rm test}}\sum_{i=1}^{n_{\rm test}}\sum_{j=1}^ k y_{ij} \log(p_{ij})$, where $n_{\rm test}$ is the size of the test set, $k$ is the number of groups,  $\log$ is the natural logarithm, $y_{ij}$ is $1$ if observation $i$ belongs to class $j$ and $0$ otherwise; $p_{ij}$ is the predicted probability that observation $i$ belongs to class $j$. Figure \ref{fig:brainloglos} indicates that we can gain efficiency by incorporating LP-learned features into classification model-building process. This aspect of our modeling makes it more powerful than merely testing the hypothesis. The user can use the output of our algorithm to upgrade the baseline predictive model. For discussion on \texttt{R} implementation, see Supplementary Materials S13.

\begin{figure}[t]
\centering
\includegraphics[width=0.5\textwidth,keepaspectratio,trim=2cm 1.5cm 2cm 2cm]{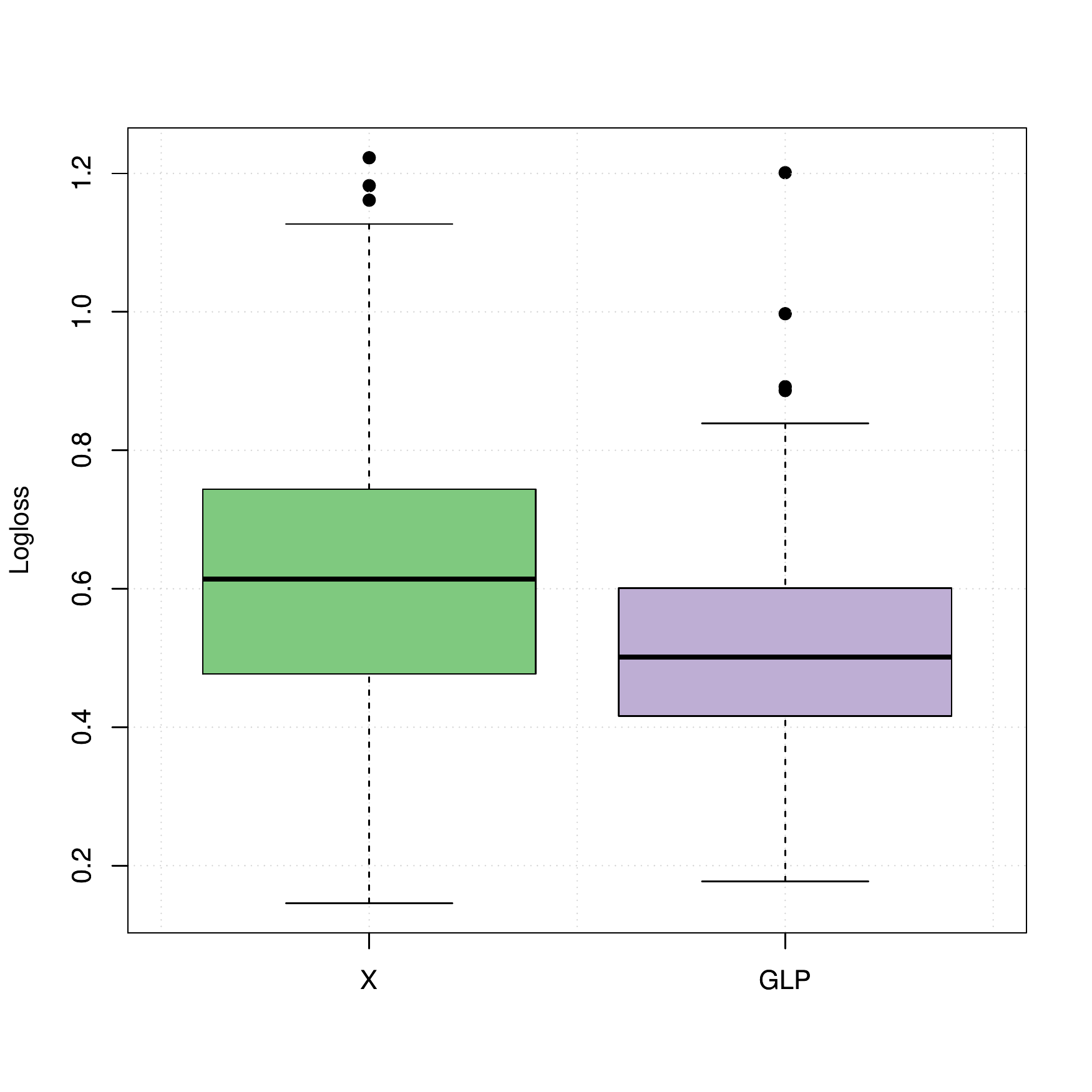}
\vskip.5em
\caption{Logarithmic loss of the multinomial logistic regression with $k=5$ classes, based on (left) original data matrix $X$ vs. (right) data-adaptive LP-transform matrix $\mathcal{T}$. }
\label{fig:brainloglos}
\end{figure}
\section{Discussion}
\label{sec:conclude}
This article provides some new graph-based modeling strategies for designing a nonparametric $k$-sample learning algorithm that is robust, increasingly automatic, and continues to work even when the dimension of the data is larger than the sample size. More importantly, it comes with an exploratory interface, which not only provides more insight into the problem but also can be utilized for developing a better predictive model at the next phase of data-modeling upon rejecting $H_0$. And for that reason, we call it a $k$-sample statistical learning problem--beyond simple testing, incorporating three modeling cultures: confirmatory, exploratory and predictive. To sum up, we must say that in designing the proposed algorithm, our priority has been flexibility as well as interpretation (D1-D5), in hopes of making it readily usable for applied data scientists. Our success comes from merging the strength of modern nonparametric statistics with the spectral graph analysis.
\section*{acknowledgement}
The authors would like to thank the editor, associate editor, and reviewers for their insightful comments. We dedicate this paper to Jerome H. Friedman on the occasion of his 80th birthday and in recognition of his pioneering work Friedman \& Rafsky(1979), which inspired this research.

\section*{Supplementary material}
\label{SM} %Supplementary material available at Biometrika online includes all remaining proofs, the
Supplementary material available at \Bka\ online includes additional numerical and theoretical discussions. All datasets and the computing codes are available via open source \texttt{R}-software package \texttt{LPKsample}, available online at \mbox{https://cran.r-project.org/package=LPKsample}.
\appendix
\appendixone
\section*{Appendix}
The Appendix section contains proofs of the main Theorems and some additional remarks on the methodological front.
%%%%%%%%%%%%%%%%%%%%%%%%%%%%%%%%%%%%%%
\subsection{Proof of Theorem 1}
\label{appendix:proof1}
\vskip.5em
 Recall $Y$ is binary with $\tp_Y(0)=n_1/n$ and $\tp_Y(1)=n_2/n$. Our goal is to find an explicit expression for the 
\[T_1(y;\wtF_Y)~=~\dfrac{\sqrt{12}\big\{\tFm_Y(y) - 1/2\big\}}{\sqrt{1-\sum_{y\in \mathscr{U}}\widetilde{p}_Y^3(y)}}.\]
We start by deriving the expression for the mid-distribution transform $\tFm_Y(y)= \wtF_Y(y)-\frac{1}{2} \tp_Y(y)$:
\beq \tFm_Y(y_i) ~=~ \left\{ \begin{array}{rl}
 -\dfrac{n_1}{2n} &\mbox{for $y_i=0$} \\[.88em]
 1-\dfrac{n_2}{2n} &\mbox{for $y_i=1$.~~~~~~~~~~~~~\,}
       \end{array} \right. \eeq
Next we determine the re-normalizing factor $1-\sum_{y\in 0,1}\widetilde{p}_Y^3(y)=3n_1n_2/n^2$. Combining previous two results we obtain the empirical LP-basis for $Y$ as
\beq T_1(y_i;\wtF_Y) ~=~ \left\{ \begin{array}{rl}
 -\sqrt{\dfrac{n_2}{n_1}} &\mbox{for $i=1,\ldots,n_1$} \\[.88em]
 \sqrt{\dfrac{n_1}{n_2}} &\mbox{for $i=n_1+1,\ldots,n$.~~~~~~~~~~~~~\,}
       \end{array} \right. \eeq
For $X$ continuous, we now aim to derive its first two empirical LP-basis. As we will see these basis functions have a direct connection with ranks. Now note that $\tFm_X(x_i)=\frac{R_i}{n} - \frac{1}{2n}$, where $R_i={\rm rank}(x_i)$. Hence we immediately have
\beq 
T_1(x_i;\wtF_X) = \sqrt{\dfrac{12}{n^2-1}} \Big(R_i - \dfrac{n+1}{2}\Big).
\eeq
This matches with the expression of $T_1(x;\wtF_X)$ as given in \eqref{eq:xLP}, up to a negligible factor. Perform Gram-Schmidt orthonormalization of $\{T_1(x;\wtF_X), T_1^2(x;\wtF_X)\}$ to obtain the second empirical LP-basis of $X$. By routine calculation, we have 
\beq
T_2(x_i;\wtF_X)\,=\,6\sqrt{5}\left\{ \big( \tFm_X(x_i) - 1/2\big)^2 - \frac{1}{12}\right\}, ~~~~\text{for}\,~ i=1,\ldots,n.
\eeq
Substituting the mid-transform function yields the desired result \eqref{eq:xLP}. This completes the proof. \qed
%%%%%%%%%%%%%%%%%%%%%%%%%%%%%%%%%%%%%%
\subsection{Proof of Theorem 2}
\label{appendix:proof2}
Applying Definition \ref{def:lpc}, we have
\[ \hLP[1,1;Y,X]\,=\,n^{-1}\sum_{i=1}^n  T_1(y_i;\wtF_Y) T_1(x_i;\wtF_X).\]
Substitute the expressions for the empirical LP-basis functions from Theorem 1 to verify that
\beq \label{e:proof2}
\hLP[1,1;Y,X]\,=\, \dfrac{\sqrt{12}}{n^2} \left\{ -\sqrt{\frac{n_2}{n_1}}\, \sum_{i=1}^{n_1} \Big( R_i - \frac{n+1}{2} \Big) \,+\,  \sqrt{\frac{n_1}{n_2}}\,  \sum_{i=n_1+1}^n \Big( R_i - \frac{n+1}{2} \Big)\right\},
\eeq
which after some algebraic manipulation can be re-written as
\beq \label{e:proof2-2}
\hLP[1,1;Y,X]\,=\, \dfrac{\sqrt{12}}{n^2 \sqrt{n_1n_2}}  \left\{ n_1  \sum_{i=n_1+1}^n R_i - n_2\Big (  \frac{n(n+1)}{2} -   \sum_{i=n_1+1}^n R_i\Big)\right\}.
\eeq
Complete the proof by noting that \eqref{e:proof2-2} is in fact 
\[\sqrt{n}\hLP[1,1;Y,X]  ~= ~\sqrt{\dfrac{12}{n\, n_1 n_2}} \left\{ \sum_{i=n_1+1}^{n} R_i\,-\, \dfrac{n_2(n+1)}{2}\right\},\]
which is equivalent to the standardized Wilcoxon statistic up to a negligible factor $\sqrt{\frac{n+1}{n}}$.

To derive the LP-representation of Mood statistic, we start with
\[ \hLP[1,2;Y,X]\,=\,n^{-1}\sum_{i=1}^n T_1(y_i;\wtF_Y) T_2(x_i;\wtF_X).\]
Proceeding as before we have
\beq \label{e:proof2-3}
\hLP[1,2;Y,X]={\small \dfrac{6\sqrt{5}}{n^3}\left[-\sqrt{\frac{n_2}{n_1}}\sum_{i=1}^{n_1}\Big\{  \Big( R_i - \frac{n+1}{2}\Big)^2 -   \frac{\sqrt{5}}{2} \Big\}+\sqrt{\frac{n_1}{n_2}}\sum_{i=n_1+1}^n \Big\{\Big( R_i - \frac{n+1}{2} \Big)^2 -\frac{\sqrt{5}}{2} \Big\}\right]}.
\eeq
Routine calculations show that \eqref{e:proof2-3} can be reduced to
\beq \sqrt{n}\hLP[1,2;Y,X] ~= ~\sqrt{\dfrac{180}{n^3\, n_1 n_2}} \sum_{i=n_1+1}^{n}\left\{ \left( R_i - \dfrac{n+1}{2} \right)^2\,-\, \dfrac{n^2+2}{12}             \right\},
\eeq
which is equivalent to the Mood statistic up to an asymptotically negligible factor. This proves the claim. \qed
%%%%%%%%%%%%%%%%%%%%%%%%%%%%%%%%%%%%%%
\subsection{Proof of Theorem 5}
\label{appendix:proof5}

Since under independence sample LP-comeans has the following weighted-average representation (where weights are marginal probabilities)
\beq \label{eq:slp}
\hLP[j,\ell;Y,Z]~=\,\mathop{\sum\sum}_{1 \le i_1,i_2 \le k} \widetilde{p}(i_1;Y)  \widetilde{p}(i_2;Z)\ T_j(i_1;\wtF_Y) T_\ell(i_2;\wtF_Z), ~~\text{$j,\ell \in \{1,\ldots,k-1\}$,}
\eeq
it is straightforward to show that in large samples the $\hLP[j,\ell;Y,Z]$ are independent and normally distributed by confirming $|\sqrt{n}\hLP[j,\ell;Y,Z]|^2$ is the score statistic for testing $H_0:\LP[j,\ell;Y,Z]=0$ against $\LP[j,\ell;Y,Z] \neq 0$. 
%%%%%%%%%%%%%%%%%%%%%%%%%%%%%%%%%%%%%%
\subsection{Additional Remarks}
\label{appendix:rems}
\vskip.35em
Remark R1.~ \textit{Spectral relaxation}: Spectral clustering converts the intractable discrete optimization problem of graph partitioning \eqref{eq:ncut3} into a 
computationally manageable eigenvector problem \eqref{eq:ncutRelax}. However, the eigenvectors of the Laplacian matrix $U_{n\times k}$ will not in general be the desired piecewise constant form \eqref{eq:ncut2}. Thus, naturally, we seek a piecewise constant matrix $\widetilde{\Psi}$ closest to the ``relaxed'' solution $U$, up to a rotation by minimizing the following squared Frobenius norm: $\| UU^T - \widetilde{\Psi}\widetilde{\Psi}^T\|^2_F$. In an important result \citet[Theorem 2]{zhang2008multiway} showed that minimizing this cost function is equivalent to performing k-means clustering on the rows of $U$. This justifies why spectral clustering scheme is the closed tractable solution of the original NP-hard normalized cut problem, as described in Section \ref{sec:ncut}. 
\vskip.45em
Remark R2.~ \textit{Computational complexity}:
An apparent limitation of all graph-based methods is that the runtime scales weakly with the sample size $n$; see Supplementary section S12. It seems worthwhile for future research to examine how to increase the speed of graph-based methods. The techniques of  \cite{deepgraph2019} may be very useful in doing this.
\vskip.45em
Remark R3.~ \textit{Joint covariate balance}: As a reviewer pointed out, researchers can use the proposed method for evaluating covariate balance in causal inference. The proposed GLP technology could be perfectly suitable for this task because:  (i) it can tackle multivariate mixed data that is prevalent in observational studies, (ii) it goes beyond the simple mean difference and captures the distributional balance across exposed and non-exposed groups, and finally, (iii) it informs causal modeler the nature of imbalance--how the distributions are different within a principal stratum, which can be used to \textit{upgrade} the logistic regression-based propensity-score matching algorithm (cf. Section \ref{sec:brain}) to reduce the bias.

%As suggested by a reviewer, it is also interesting to investigate

\end{document}